\documentclass[12pt]{scrartcl}   
\usepackage{graphicx}
\usepackage{enumerate}
\usepackage{amssymb, amsmath, amsthm}

\newtheorem{satz}{Theorem}[section]
\newtheorem{lem}[satz]{Lemma}
\newtheorem{proposition}[satz]{Proposition}

\newtheorem{corollary}[satz]{Corollary}

\newtheorem{anm}[satz]{Remark}

\newtheorem{defi}{Definition}[section]

\usepackage[T1]{fontenc}
\usepackage{lmodern}

\newcommand{\N}{\ensuremath{{\mathbb N}}}

\newcommand{\E}{\ensuremath{{\mathbb E}}}
\DeclareMathOperator*{\esssup}{esssup}

\begin{document}

\title{Optimal multiple stopping with random waiting times}
\author{S\"oren Christensen\thanks{Christian-Albrechts-Universit\"at, Mathematisches Seminar, Ludewig-Meyn-Str. 4, 24098 Kiel, Germany, email: \textit{lastname}@math.uni-kiel.de.} \and Albrecht Irle\footnotemark[1] \and Stephan J\"urgens\footnotemark[1]}
\date{\today}
\maketitle

\begin{abstract}
In the standard models for optimal multiple stopping problems it is assumed that between two exercises there is always a time period of deterministic length $\delta$, the so called refraction period. This prevents the optimal exercise times  from bunching up together on top of the optimal stopping time for the one-exercise case. In this article we generalize the standard model by considering random refraction times. We develop the theory and reduce the problem to a sequence of ordinary stopping problems  thus extending the results for deterministic times. This requires an extension of the underlying filtrations in general. Furthermore we consider the Markovian case and treat an example explicitly.
\end{abstract}


\section{Introduction}\label{intro}
The aim of this paper is to extend the theory of optimal multiple stopping. These stopping problems with multiple exercise opportunities arise in different fields of applied probability, in particular in the analysis of options traded on the energy market. For example swing options entitle the buyer to exercise a certain right $n$ times in a given time interval. The buyer of the option is faced with the following optimization problem: What are the best times to exercise these rights? This leads to the problem
\[\mbox{Maximize }\E(Y(\tau_1)+...+Y(\tau_n)) \mbox{ for $(\tau_1,...,\tau_n)$ with }  \tau_1\leq ...\leq \tau_n.\]
Without any further restrictions one can see that the optimal strategy is to exercise all rights at that same time, namely at the time when one would exercise for $n=1$. But typically there are restrictions imposed on the exercise dates. The most relevant restriction is that between each two exercise times there must be a pre-specified time interval, the so called refraction period, of length $\delta$, i.e. $\tau_i+\delta\leq \tau_{i+1}$ for all $i$.\\
In the existing models for such situations it is assumed that $\delta$ is a given constant, see e.g. \cite{Touzi}, but in real world situations random waiting times might arise. So, in this article we introduce random waiting times $\delta_1,...,\delta_{n-1}$. For dealing with this extended problem we note that the optimal strategies $\tau_i,i\geq 2,$ will depend on the values $\delta_1,...,\delta_{i-1}$; therefore we have to enrich the standard filtration to include the information given by the random waiting times. This is carried out for the continuous time case in Section \ref{model} and first properties of the enriched filtration are noted there.\\
The main tool for a solution of optimal multiple stopping problems lies in the reduction of the original problem to a sequence of $n$ ordinary optimal stopping problems. This approach is well known for the problem with deterministic waiting times, see e.g. \cite[Section 2]{Bender} and the references therein for the discrete time case and \cite[Proposition 3.2]{Car} and \cite[Theorem 2.1]{Touzi} for the continuous time case. An extension to a more general situation can be found in \cite{KQR}. This reduction can be used to apply known techniques for ordinary optimal stopping problems to solve multiple stopping problems (semi-)explicitly or numerically, see e.g. \cite{Jaillet}, \cite{Thompson} and \cite{Meinshausen}, also \cite{BenderDual} and \cite{BenderDual2}. We establish such a reduction principle in Section \ref{reduction} and furthermore carry over the results to the discrete time case. In Section \ref{markov} we establish the theory for underlying Markov processes. We treat an example in Subsection \ref{subsection:house}, where we can find the explicit solution. \\
In Section \ref{sec:second_model} we introduce a second model that covers other classes of real-world problems, such as employee options. The basic idea is that additionally to the reward process another process is running, and it is only possible to exercise again when the second process enters a given set $B$. We develop the theory, but because this model is easier to handle and many ideas are similar to the arguments given before we just sketch the proofs.
We use these results for treating an example with explicit solution in Section \ref{sec:example}. This is remarkable since in the theory for deterministic waiting times no such examples seem to be known.

\section{The model and first properties}
\label{model}
For the following we fix a probability space $(\Omega,\mathcal{F},\mathbb{P})$ and a filtration $\mathfrak{F}=(\mathcal{F}_t)_{t \geq 0}$ fulfilling the usual conditions, i.e. $\mathfrak{F}$ is right-continuous and $\mathcal{F}_0$ contains all $P$-null-sets. Furthermore let $ Y = (Y(t))_{t \geq 0}$ be an $\mathfrak{F}$-adapted,
    non-negative and right-continuous stochastic process with the property
    \begin{equation}\label{eigenschaft supremumsprozess int}
            \E\Bigr{(} \sup_{t\geq 0}Y(t)\Bigr{)}< \infty.
    \end{equation}
    We denote by $\mathcal{S}$ the set of all $\mathfrak{F}$-stopping times with values in $[0, \infty]$ and write
     \[
        \mathcal{S}_\sigma:=\{ \tau \in \mathcal{S};\; \tau \geq \sigma\}
    \]
    for all $\sigma\in\mathcal{S}$. Note that the value $\infty$ is admitted.
    Furthermore, we write $\mathcal{F}_\infty:= \sigma\Bigr{(}\bigcup_{t\geq 0} \mathcal{F}_t \Bigl{)}$
    and $Y(+\infty):= \limsup  Y(t).  $\\
Moreover we fix a number $n$ that represents the number of exercise opportunities. As discussed in the introduction, between each two exercise times $\tau_i$ and $\tau_{i+1}$ we have to wait at least $\delta_i$ time units. Now we assume $\delta_1,...,\delta_{n-1}$ to be non-negative random variables which are finite a.s. In this case the holder of the option has to use the information given by the waiting times for the next decision. This information is a priori unrelated to the information given by the filtration $\mathfrak{F}$. Hence $\mathfrak{F}$ is not the adequate filtration for formulating the multiple optimal stopping problem with random waiting times:\\
If the holder of the option exercises the first time using the strategy $\tau_1$ and waits for $\delta_1$ time units, then the information available at time $t$ is given by $\mathcal{F}_t \cup \{\{ \tau_1+\delta_1 \leq s\};\;s\leq t\}$. Hence for the strategy $\tau_2$  the holder has this information  at time $t$, i.e. $\tau_2$ should be modeled as a stopping time with respect to $(\sigma(\mathcal{F}_t \cup \{\{ \tau_1+\delta_1 \leq s\};\;s\leq t\}))_{t\geq 0}$. In the following subsection we study this filtration in detail before formulating the multiple optimal stopping problem with random waiting times in this model.
In order to enhance readability, some proofs are are given in an appendix.

\subsection{Exercise strategies with random waiting times}\label{abschnitt zuf wart filt}

For shorter notation we use the following definition:

\begin{defi}\label{defi filtration der stoppzeit}
        For each $\rho:\Omega \rightarrow [0,\infty]$ let $\mathcal{G}^\rho$ be the smallest filtration such that $\rho$ is a stopping time, i.e.
        \[
            \mathcal{G}^\rho_t = \bigcap \{\mathcal{A}_t;\; \mathfrak{A}=(\mathcal{A}_s)_{s\geq 0} \text{ is a filtration and $\rho$ is an $\mathfrak{A}$-stopping time}\},\hspace{0.5cm}{t\geq 0}.
        \]
\end{defi}

For example if $\rho$ is an $\mathcal{F}$-stopping time, then $\mathcal{G}_t^\rho\subseteq \mathcal{F}_t$ for all $t\geq 0$. In particular if $\rho$ is a constant, then $\mathcal{G}_t^\rho = \{\emptyset,\Omega\}$. Furthermore one immediately checks that
    \begin{eqnarray}
        \mathcal{G}^\rho_t
        & = &  \sigma\Bigr{(} \{ \{\rho\leq s\};\; s\leq t\} \Bigl{)}\label{eigenschaft charak neue filtration 1} \\
        & = & \sigma\Bigr{(} \rho \, 1_{[0,t]}(\rho),\, 1_{\{0\}}(\rho) \Bigl{)}. \label{eigenschaft charak neue filtration 2}
    \end{eqnarray}

\begin{defi}\label{defi filtration der stoppzeit 2}
For each filtration $\mathfrak{A}=(\mathcal{A}_t)_{t\geq 0}$ and each $\rho:\Omega \rightarrow [0,\infty]$ let $\mathfrak{A}^\rho= (\mathcal{A}_t^\rho)_{t\geq 0}$ be the smallest right-continuous filtration containing $\mathfrak{A}$ such that $\rho$ is a stopping time, i.e. $\mathfrak{A}^\rho$ the right-continuous filtration generated by filtration  $\Bigr{(}\sigma(\mathcal{A}_t,\mathcal{G}^\rho_t)\Bigl{)}_{t\geq 0}$.\\
Furthermore for all  $\rho_1,...,\rho_n:\Omega \rightarrow [0,\infty]$ define recursively
\[\mathfrak{A}^{\rho_1,...,\rho_i}= \bigr{(}\mathfrak{A}^{\rho_1,...,\rho_{i-1}}\bigl{)}^{\rho_i}.\]
\end{defi}

The following proposition lists some properties of $\mathfrak{A}^\rho$ for later use.

\begin{proposition}\label{satz eigenschaften von neuer filtration}
    With the notation given above it holds that:
    \begin{enumerate}[(i)]
        \item $\mathfrak{A}^\rho=\mathfrak{A}$, if $\rho$ is an $\mathfrak{A}$-stopping time and $\mathfrak{A}$ is right-continuous.
        \item $\mathcal{G}^\rho_{t-s}= \mathcal{G}^{\rho+s}_{t}$ and $\mathcal{A}^\rho_{t-s}\subseteq \mathcal{A}_t^{\rho+s}$  for all $s\leq t$.
        \item  $\mathcal{A}_t^{\sigma+\rho}\mid_{ \{\sigma=\tau\}} =  \mathcal{A}_t^{\tau+\rho}\mid_{ \{\sigma=\tau\}}$ for all $t\geq 0$ and $\mathfrak{A}$-stopping times $\sigma$ and $\tau$, where $\mathcal{C}\mid_{ D}:= \{ C\cap D;\; C\in\mathcal{C}\}$ for all $\sigma$-algebras $\mathcal{C}$, $D\subseteq\Omega$.
        \item $\mathcal{A}^{\sigma + \rho}_t \supseteq \mathcal{A}^{\tau + \rho}_t$ for all $t\geq 0$ and $\mathfrak{A}$-stopping times $\sigma\leq \tau$ with countable range.
    \end{enumerate}
\end{proposition}

\begin{proof}
See \ref{proof satz eigenschaften von neuer filtration}.
\end{proof}

Now we come back to the model described at the beginning of this section by introducing the strategies:
\begin{defi}\label{defi stoppzeiten zufaellig} Write $\delta:=(\delta_1,...,\delta_{n-1})$ for short. For each $\mathfrak{F}$-stopping time $\sigma$ define
        \begin{align}
            \mathcal{S}_\sigma^n(\delta,\mathfrak{F}):=\Bigr{\{}(\tau_1,...,\tau_n);\;& \tau_1 \text{ is $\mathfrak{F}$-stopping time and } \sigma\leq \tau_1 ,
            \\ & \tau_i \text{ is $\mathfrak{F}^{\tau_1+\delta_1,...,\tau_{i-1}+\delta_{i-1}}$-stopping time } \nonumber\\
            & \text{and } \tau_{i-1}+\delta_{i-1}\leq \tau_i, i=2,...,n\nonumber \Bigl{\}}
        \end{align}
        and
        \begin{eqnarray}
         \mathcal{S}_\sigma^n(\delta,\mathfrak{F})_{\text{disc}}:=\Bigr{\{}(\tau_1,...,\tau_n)\in  \mathcal{S}_\sigma^n(\delta,\mathfrak{F});\; \text{range}(\tau_1) \, \text{is countable}\Bigl{\}}.
        \end{eqnarray}
\end{defi}
We also write $\mathcal{S}_\sigma^n(\delta):= \mathcal{S}_\sigma^n(\delta,\mathfrak{F})$ and $\mathcal{S}_\sigma^n(\delta)_{\text{disc}}:=\mathcal{S}_\sigma^n(\delta,\mathfrak{F})_{\text{disc}}$ for short. The notation is a direct generalization of that given, e.g., in \cite{Touzi} for deterministic $\delta$ since from Proposition 2.1 it follows
\[\mathcal{S}_\sigma^n(\delta)= \{ \tau\in\mathcal{S}^n;\; \sigma\leq \tau_1,\; \tau_{i-1}+\delta_{i-1}\leq \tau_i, i=1,..,n\}\]
if $\delta_1,...,\delta_{n-1}$ are $\mathfrak{F}$-stopping times, in particular if $\delta=\delta_1=...=\delta_{n-1}$ are constants.

The following lemma summarizes some useful facts for later use.
\begin{lem}\label{lemma eigenschaften stopp zufaellig}
    Let  $\sigma$ and $\tau$ be  $\mathfrak{F}$-stopping times and write $L\delta:=(\delta_2,...,\delta_{n-1})$.
    \begin{enumerate}[(i)]
        \item  It holds that
        \[\mathcal{S}_\sigma^n(\delta,\mathfrak{F})= \Bigr{\{} (\tau_1,...,\tau_n);\; \tau_1\in\mathcal{S}_\sigma,\, (\tau_2,...,\tau_n)\in \mathcal{S}_{\delta_1+\tau_1}^{n-1}(L\delta,\mathfrak{F}^{\tau_1+\delta_1})\Bigl{\}}\] and
          \[\mathcal{S}_\sigma^n(\delta,\mathfrak{F})_{\textnormal{disc}}= \Bigr{\{} (\tau_1,...,\tau_n);\; \tau_1\in\mathcal{S}_{\sigma,\textnormal{disc}} ,\, (\tau_2,...,\tau_n)\in \mathcal{S}_{\delta_1+\tau_1}^{n-1}(L\delta,\mathfrak{F}^{\tau_1+\delta_1})\Bigl{\}}.\]
        \item
            For all $(\tau_2,...,\tau_{n})\in \mathcal{S}_{\tau+\delta_1}^{n-1}(L\delta,\mathfrak{F}^{\tau+\delta_1})$ and $\lambda{\geq 0}$ it holds that
            \[
                \Bigr{(}  \tau_2 1_{\{\tau=\lambda\}} +  (\lambda+ \delta_1) 1_{\{\tau\neq\lambda\}},..., \tau_n 1_{\{\tau=\lambda\}} + (\lambda+\delta_1+...+\delta_{n-1}) 1_{\{\tau\neq\lambda\}}  \Bigl{)}\in\mathcal{S}_{\lambda+\delta_1}^{n-1}(L\delta,\mathfrak{F}^{\lambda+\delta_1}).
            \]
         \item
         For all $(\tau_2,...,\tau_{n})\in \mathcal{S}_{\lambda+\delta_1}^{n-1}(L\delta,\mathfrak{F}^{\lambda+\delta_1})$  it holds that
            \[
                \Bigr{(}  \tau_2 1_{\{\tau=\lambda\}} +  (\tau+ \delta_1) 1_{\{\tau\neq\lambda\}},..., \tau_n 1_{\{\tau=\lambda\}} + (\tau+\delta_1+...+\delta_{n-1}) 1_{\{\tau\neq\lambda\}}  \Bigl{)}\in \mathcal{S}_{\tau+\delta_1}^{n-1}(L\delta,\mathfrak{F}^{\tau+\delta_1}).
            \]
          \item
              For all $(\tau_2,...,\tau_{n})\in \mathcal{S}_{\tau+\delta_1}^{n-1}(L\delta,\mathfrak{F}^{\tau+\delta_1})$  and $c{\geq 0}$
              \[
                (\tau_2+c,...,\tau_{n}+c)\in\mathcal{S}_{\tau+ c +\delta_1}^{n-1}(L\delta,\mathfrak{F}^{\tau+c+\delta_1}).
             \]
                 \end{enumerate}
\end{lem}

\begin{proof}
See \ref{proof lemma eigenschaften stopp zufaellig}.
\end{proof}

\begin{anm}\label{bem komp max un min bildung}
At a first glance one might think that the set $\mathcal{S}_\sigma^n(\delta,\mathfrak{F})$ is closed under taking (component wise) maxima. This is not true in general: Indeed
if $\rho_1$ and $\tau_1$ are $\mathfrak{F}$-stopping times with $\rho_1 \geq \tau_1$, $\tau_2$ is an $\mathfrak{F}^{\tau_1+\delta_1}$-stopping time and $\rho_2$ an $\mathfrak{F}^{\rho_1+\delta_1}$-stopping time with $\rho_2\leq \tau_2$ and $\tau_2$ is not an $\mathfrak{F}^{\rho_1+\delta_1}$-stopping time, then $\rho_2\vee \tau_2=\tau_2$ is not an $\mathfrak{F}^{(\rho_1\vee\tau_1)+\delta_1 }= \mathfrak{F}^{\rho_1+\delta_1}$-stopping time.\\
On the other hand if $\delta_1,...,\delta_{n-1}$ are $\mathfrak{F}$-stopping times, then this property obviously holds.
\end{anm}

\subsection{Formulation of the problem}\label{Abschnitt zuf wart}

Using the notation given in the previous subsection we can state the problem of optimal multiple stopping with random waiting times:\\
Maximize the expectation
\[\E\Bigr{(}\sum_{i=1}^nY(\tau_i)\Bigl{)}\]
over all $(\tau_1,...,\tau_n)\in  \mathcal{S}_0^n(\delta,\mathfrak{F})$. To treat this problem we extend it in the usual way as follows:

\begin{defi}\label{defi stoppproblem zuf}
     For each $\mathfrak{F}$-stopping time $\sigma$ write
     \begin{equation}\label{eigenschaft defi stoppproblem mit zuf wart}
        Z_n(\sigma):= Z_n^{\delta,\mathfrak{F}}(\sigma):= \underset{\tau\in \mathcal{S}_\sigma^n(\delta,\mathfrak{F})}{\esssup} \E\Bigr{(}\sum_{i=1}^nY(\tau_i)\Big{|}\mathcal{F}_\sigma\Bigl{)}
     \end{equation}
     and
     \begin{equation}
        Z_n(\sigma)_{\text{disc}}:= Z_n^{\delta,\mathfrak{F}}(\sigma)_{\text{disc}}:= \underset{\tau\in \mathcal{S}_\sigma^n(\delta,\mathfrak{F})_{\text{disc}}}{\esssup} \E\Bigr{(}\sum_{i=1}^nY(\tau_i)\Big{|}\mathcal{F}_\sigma\Bigl{)}.
     \end{equation}
      Furthermore let $Z_0(\sigma)\equiv Z_0(\sigma)_{\text{disc}}:\equiv 0$.
\end{defi}

We immediately obtain

\begin{proposition}\label{satz zus zuf wart disk nicht disk}
    For all $\mathfrak{F}$-stopping times $\sigma$ it holds that
    \begin{enumerate}[(i)]
        \item
            $Z_n(\sigma)_{\textnormal{disc}}\leq Z_n(\sigma)$.
        \item
            If $\delta_1,...,\delta_{n-1}$ are $\mathfrak{F}$-stopping times, then there is equality in $(i)$.
    \end{enumerate}
\end{proposition}

\begin{proof}
 $(i)$ is immediate since $\mathcal{S}_\sigma^n(\delta,\mathfrak{F})_{\text{disc}} \subseteq \mathcal{S}_\sigma^n(\delta,\mathfrak{F})$.\\
 For $(ii)$ let $\tau\in \mathcal{S}_\sigma^n(\delta,\mathfrak{F})$. Take a sequence $(\tau_k^1)_{k\in\N}$ of $\mathfrak{F}$-stopping times with countable range  and $\tau_k^1 \downarrow \tau_1 $ for $k\rightarrow \infty$. For each $k\in\N$ define
                \begin{equation}\label{eigenschaft tilde tau}
                   \tilde{\tau}_k^1:=\tau_k^1 \quad \text{and} \quad
                    \tilde{\tau}_k^{i+1}= \max\{  \tilde{\tau}_k^i+\delta_i, \tau_{i+1} \}.
                \end{equation}
                Then $(\tilde{\tau}_k^1,...,\tilde{\tau}_k^n)\in\mathcal{S}_\sigma^n(\delta)_{\text{disc}}$ for all $k\in\N$ and $(\tilde{\tau}_k^i)_{k\in\N}$ converges to $\tau_i$ a.s. for all $i\in\N_{\leq n}$. Therefore
                \[
                    \E\Bigr{(}\sum_{i=1}^nY(\tau_i)\Big{|}\mathcal{F}_\sigma\Bigl{)}
                    = \lim_{k\rightarrow \infty } \E\Bigr{(}\sum_{i=1}^nY(\tilde{\tau}_k^i)\Big{|}\mathcal{F}_\sigma\Bigl{)}
                    \leq  \underset{\rho\in \mathcal{S}_\sigma^n(\delta)_{\text{disc}}}{\esssup}\! \E\Bigr{(}\sum_{i=1}^nY(\rho_i)\Big{|}\mathcal{F}_\sigma\Bigl{)}\leq Z_n(\sigma)_{\text{disc}}.
                \]
\end{proof}


\section{Reduction principle for the problem}\label{reduction}

To develop the theory of multiple optimal stopping the following lemma is fundamental:

\begin{lem}\label{lemma oben gerich} The sets
    \[ \left\{\E\Bigr{(}\underset{i=1}{\overset{n}{\sum}}Y(\tau_i)\Big{|} \mathcal{F}_\sigma\Bigl{)};\; \tau\in\mathcal{S}_\sigma^{n}(\delta)\right\} \quad \text{and} \quad  \left\{\E\Bigr{(}\underset{i=1}{\overset{n}{\sum}}Y(\tau_i)\Big{|} \mathcal{F}_\sigma\Bigl{)};\; \tau\in\mathcal{S}_\sigma^{n}(\delta)_{\textnormal{disc}}\right\}
 \]
are directed upwards; here an ordered set $M$ is called directed upwards if for all $a,b\in M$ there exists $c\in M$ such that $\max\{a,b\}\leq c$.
\end{lem}

\begin{proof}
See \ref{proof lemma oben gerich}
\end{proof}

Now we can formulate the first step in the reduction of the problem:

\begin{satz}\label{sazu zuf wart ind schr}
    For each $\mathfrak{F}$-stopping time $\sigma$ it holds that
    \begin{equation}
         Z_{n}^{\delta,\mathfrak{F}}(\sigma)= \underset{\tau\in \mathcal{S}_\sigma}{\esssup} \; \E\Bigr{(} Y(\tau)+ \E\bigr{(}Z_{n-1}^{L\delta,\mathfrak{F}^{\tau+\delta_1}}(\tau+\delta_1)|\mathcal{F}_\tau \bigl{)}
                \Big{|}\mathcal{F}_\sigma \Bigl{)}
    \end{equation}
    and
    \begin{equation}\label{eigenschaft ind schritt zuf wart disk fall}
        Z_{n}^{\delta,\mathcal{F}}(\sigma)_{\textnormal{disc}}= \underset{\tau\in \mathcal{S}_{\sigma,\textnormal{disc}}}{\esssup} \; \E\Bigr{(} Y(\tau)+ \E\bigr{(}Z_{n-1}^{L\delta,\mathcal{F}^{\tau+\delta_1}}(\tau+\delta_1)|\mathcal{F}_\tau \bigl{)}  \Big{|}\mathcal{F}_\sigma \Bigl{)}.
    \end{equation}
\end{satz}

\begin{proof}We only prove the first statement, the second one may be proved in the same manner.\\
Let  $(\tau_1,...,\tau_n)\in\mathcal{S}_\sigma^n(\delta,\mathfrak{F})$.  By Lemma \ref{lemma eigenschaften stopp zufaellig} (i) it holds that $\tau_1\in\mathcal{S}_\sigma$ and $(\tau_2,...,\tau_n)\in \mathcal{S}_{\delta_1+\tau_1}^{n-1}(L\delta,\mathfrak{F}^{\tau_1+\delta_1})$. Therefore
            \begin{eqnarray*}
                \E\Bigr{(}\sum_{i=1}^{n}Y(\tau_i)\Big{|}\mathcal{F}_{\sigma}\Bigl{)}
                & = & \E\Bigr{[} Y(\tau_1) + \E\Bigr{(}\sum_{i=2}^{n}Y(\tau_i)\Big{|}\mathcal{F}_{\tau_1+\delta_1}\Bigl{)}\Big{|}\mathcal{F}_\sigma\Bigl{]}\\
                & \leq &  \E \Bigr{(} Y(\tau_1) + Z_{n-1}^{L\delta,\mathfrak{F}^{\tau_1+\delta_1}}(\tau_1+\delta_1) \Big{|}\mathcal{F}_\sigma\Bigl{)}\\
                & = &  \E \Bigr{[} Y(\tau_1) + \E\Bigr{(} Z_{n-1}^{L\delta,\mathfrak{F}^{\tau_1+\delta_1}}(\tau_1+\delta_1)\Big{|} \mathcal{F}_{\tau_1}\Bigl{)} \Big{|}\mathcal{F}_\sigma\Bigl{]}\\
                & \leq & \underset{\tau\in\mathcal{S}_\sigma}{\esssup}\; \E\Bigr{(} Y(\tau)+ \E\bigr{(}Z_{n-1}^{L\delta,\mathfrak{F}^{\tau+\delta_1}}(\tau+\delta_1)|\mathcal{F}_\tau \bigl{)}  \Big{|}\mathcal{F}_\sigma \Bigl{)}.
            \end{eqnarray*}
To prove the other inequality let $\tau_1\in \mathcal{S}_\sigma$.
Since the set
\[
    \{\E\Bigr{(}\underset{i=1}{\overset{n}{\sum}}Y(\tau_i)\Big{|} \mathcal{F}_\sigma\Bigl{)};\; \tau\in\mathcal{S}_\sigma^{n}(\delta)\}\]
is directed upwards by Lemma \ref{lemma oben gerich}, there exists a sequence $(\tau^k)_{k\in\N}\in(\mathcal{S}_\sigma^{n}(\delta))^\N$ such that
\[
    Z_n^{\delta,\mathfrak{F}}(\sigma)=
    \lim\limits_{k \rightarrow \infty} {\uparrow}\E\Bigr{(}\sum_{i=1}^nY(\tau_i^k)\Big{|}\mathcal{F}_\sigma\Bigl{)}
\]
see e.g. \cite[Lemma 1.3]{PS}.
             Since $(\tau_1,\tau_2^k,...,\tau_{n}^k)\in \mathcal{S}_\sigma^n(\delta,\mathfrak{F})$ for all $k\in\N$ by Lemma \ref{lemma eigenschaften stopp zufaellig} (i) we obtain
            \begin{eqnarray*} \qquad\qquad\qquad\;\,
                Z_{n}^{\delta,\mathfrak{F}}(\sigma) & \geq & \limsup_{k\rightarrow \infty} \; \E\Bigr{(}Y(\tau_1) + \sum_{i=2}^{n}Y(\tau_i^k)\Big{|}\mathcal{F}_{\sigma}\Bigl{)}\\
                & = & \limsup_{k\rightarrow \infty} \; \E\Bigr{[}Y(\tau_1) + \E\Bigr{(}\sum_{i=2}^{n}Y(\tau_i^k)\Big{|}\mathcal{F}_{\tau_1+\delta}\Bigl{)}\Big{|} \mathcal{F}_\sigma \Bigl{]}\\
                & = &
                \E\Bigr{[}Y(\tau_1) + \lim_{k\rightarrow \infty} \uparrow \E\Bigr{(}\sum_{i=2}^{n}Y(\tau_i^k)\Big{|}\mathcal{F}_{\tau_1+\delta}\Bigl{)}\Big{|} \mathcal{F}_\sigma \Bigl{]}\\
                & = &  \E\Bigr{(}Y(\tau_1) + Z_{n-1}^{L\delta,\mathfrak{F}^{\tau_1+\delta_1}}(\tau_1+\delta_1) \Big{|} \mathcal{F}_\sigma \Bigl{)}\\
                & = &  \E\Bigr{(}Y(\tau_1) + \E\bigr{(} Z_{n-1}^{L\delta,\mathfrak{F}^{\tau_1+\delta_1}}(\tau_1+\delta_1)  | \mathcal{F}_{\tau_1} \bigl{)}\Big{|} \mathcal{F}_\sigma \Bigl{)}.
            \end{eqnarray*}

\end{proof}

A technical problem to use the reduction theorem given above is that the mapping $t\mapsto \E\bigr{(}Z_{n-1}^{L\delta,\mathfrak{F}^{t+\delta_1}}(t+\delta_1)|\mathcal{F}_t \bigl{)}$ does not have to be right-continuous. We overcome this problem by giving a right-continuous modification, compare  \cite{Car} . To this end we need the following lemmas:

\begin{lem}\label{lemma abz stoppteiten filtration}
    For each $\mathfrak{F}$-stopping time $\tau$ and all $\lambda{\geq 0}$ it holds that
    \[
        1_{\{\tau=\lambda\}} \, Z_{n-1}^{L\delta,\mathfrak{F}^{\tau+\delta_1}}(\tau+\delta_1) =  1_{\{\tau=\lambda\}} \, Z_{n-1}^{L\delta,\mathfrak{F}^{\lambda+\delta_1}}(\lambda+\delta_1).
    \]
\end{lem}

\begin{proof}
Let $(\tau_2,...,\tau_n)\in\mathcal{S}_{\tau+\delta_1}^{n-1}(L\delta,\mathfrak{F}^{\tau+\delta_1})$. By Lemma \ref{lemma eigenschaften stopp zufaellig}.(ii) the random vector
            \[
                (\rho_2,...,\rho_n):=\Bigr{(}  \tau_2 1_{\{\tau=\lambda\}} +  (\lambda+ \delta_1)
                 1_{\{\tau\neq\lambda\}},..., \tau_n 1_{\{\tau=\lambda\}} + (\lambda+\delta_1+...+\delta_{n-1})
                 1_{\{\tau\neq\lambda\}}  \Bigl{)}
             \]
              is an element of $\mathcal{S}_{\lambda+\delta_1}^{n-1}(L\delta,\mathfrak{F}^{\lambda+\delta_1})$.
             Using this fact together with Proposition \ref{satz eigenschaften von neuer filtration}.(iii) we obtain

            \begin{eqnarray*}
                 1_{\{\tau=\lambda\}} \, \E\Bigr{(} \sum_{i=2}^n Y(\tau_i)
                  \Big{|} \mathcal{F}^{\tau+\delta_1}_{\tau+\delta_1}\Bigl{)}
                & = &  \E\Bigr{(} \sum_{i=2}^n 1_{\{\tau=\lambda\}} \, Y(\tau_i)
                  \Big{|} \mathcal{F}^{\tau+\delta_1}_{\tau+\delta_1}\mid_{ \{\tau=\lambda\}}\Bigl{)}\\
                & = &  \E\Bigr{(} \sum_{i=2}^n 1_{\{\tau=\lambda\}} \, Y(\rho_i)
                  \Big{|} \mathcal{F}^{\lambda+\delta_1}_{\lambda+\delta_1}\mid_{ \{\tau=\lambda\}}\Bigl{)}\\
                & = &  1_{\{\tau=\lambda\}} \, \E\Bigr{(} \sum_{i=2}^n Y(\rho_i)
                  \Big{|}\mathcal{F}^{\lambda+\delta_1}_{\lambda+\delta_1}\Bigl{)}\\
                & \leq & 1_{\{\tau=\lambda\}} \,
                    Z_{n-1}^{L\delta,\mathfrak{F}^{\lambda+\delta_1}}(\lambda+\delta_1).
            \end{eqnarray*}
The reverse inequality holds by the same argument using Lemma \ref{lemma eigenschaften stopp zufaellig}.(iii).
\end{proof}

\begin{lem}\label{defilem zuf wart supermart}
    The process $\overline{Z}_{n-1}^{\delta}$ given by
    \begin{equation}\label{defi zuf wart Zoverline}
        \overline{Z}_{n-1}^{\delta}(t):=
        \E\bigr{(}Z_{n-1}^{L\delta,\mathfrak{F}^{t+\delta_1}}(t+\delta_1)|\mathcal{F}_t \bigl{)}
    \end{equation}
    is a supermartingale with the following properties:
    \begin{enumerate}[(i)]
        \item
            $\overline{Z}_{n-1}^{\delta}$ has a right-continuous modification
            $\overline{Z}_{n-1}^{\delta,r}$.
        \item For all $\mathfrak{F}$-stopping times $\tau$ with countable range it holds that
            \begin{equation}\label{eigenschaft zuf wart bed erw}
                \overline{Z}_{n-1}^{\delta,r}(\tau)= \overline{Z}_{n-1}^{\delta}(\tau)=
                \E\bigr{(}Z_{n-1}^{L\delta,\mathfrak{F}^{\tau+\delta_1}}(\tau+\delta_1)|\mathcal{F}_\tau \bigl{)}.
            \end{equation}
    \end{enumerate}
\end{lem}

\begin{proof}
      For all $s\leq t$ and $r\geq 0$ using Theorem \ref{satz eigenschaften von neuer filtration}.(iv) it holds that
      \[
        \mathcal{F}_r^{t+\delta}\subseteq \mathcal{F}_r^{s+\delta},
      \]
      hence
            \begin{eqnarray*}
                \E( \overline{Z}_{n-1}^{\delta}(t)|\mathcal{F}_s)
                & = &   \E( Z_{n-1}^{L\delta,\mathfrak{F}^{t+\delta_1}}(t+\delta_1) |\mathcal{F}_s)\\
                & \leq & \E( Z_{n-1}^{L\delta,\mathfrak{F}^{s+\delta_1}}(t+\delta_1) |\mathcal{F}_s)\\
                & = & \E\Bigr{(} \E\Bigr{[} Z_{n-1}^{L\delta,\mathfrak{F}^{s+\delta_1}}(t+\delta_1) \Big{|}\mathcal{F}_{s+\delta_1}^{s+\delta_1}\Bigl{]} \Big{|} \mathcal{F}_s \Bigl{)}\\
                & = & \E\Bigr{(} \E\Bigr{[} \underset{\tau\in \mathcal{S}_{t+\delta_1}^{n-1}(L\delta,\mathfrak{F}^{s+\delta_1})}{\esssup} \E\Bigr{(}\sum_{i=1}^{n-1}Y(\tau_i)\Big{|}\mathcal{F}_{t+\delta_1}^{s+\delta_1}\Bigl{)} \Big{|}\mathcal{F}_{s+\delta_1}^{s+\delta_1}\Bigl{]} \Big{|} \mathcal{F}_s \Bigl{)}.
            \end{eqnarray*}
            Since the set $\Bigr{\{} \E\Bigr{(}\underset{i=1}{\overset{n-1}{\sum}}Y(\tau_i)\Big{|}\mathcal{F}_{t+\delta_1}^{s+\delta_1}\Bigl{)};\; \tau\in \mathcal{S}_{t+\delta_1}^{n-1}(L\delta,\mathfrak{F}^{s+\delta_1})
             \Bigl{\}}$ is upwards directed by Lemma \ref{lemma oben gerich}, we obtain
            \begin{eqnarray*}
                 \E( \overline{Z}_{n-1}^{\delta}(t)|\mathcal{F}_s)
                & \leq  & \E\Bigr{(}  \underset{\tau\in \mathcal{S}_{t+\delta_1}^{n-1}(L\delta,\mathfrak{F}^{s+\delta_1})}{\esssup} \E\Bigr{(}\sum_{i=1}^{n-1}Y(\tau_i)\Big{|}\mathcal{F}_{s+\delta_1}^{s+\delta_1}\Bigl{)} \Big{|} \mathcal{F}_s \Bigl{)}\\
                & \leq & \E\Bigr{(}  \underset{\tau\in \mathcal{S}_{s+\delta_1}^{n-1}(L\delta,\mathfrak{F}^{s+\delta_1})}{\esssup} \E\Bigr{(}\sum_{i=1}^{n-1}Y(\tau_i)\Big{|}\mathcal{F}_{s+\delta_1}^{s+\delta_1}\Bigl{)} \Big{|} \mathcal{F}_s \Bigl{)}\\
                & = & \overline{Z}_{n-1}^{\delta}(s),
            \end{eqnarray*}
            i.e. the supermartingale property.
      \begin{description}
            \item{$\text{(i):}$} Let $t\geq 0$ and $(t_k)_{k\in\N}$ a sequence converging to $t$ from above. Since $\overline{Z}_{n-1}^{\delta}$ is a supermartingal, the limit $\lim_{k\rightarrow \infty} \E \Bigr{(} \overline{Z}_{n-1}^{\delta}(t_k) \Bigl{)}$ exists and is bounded above by  $ \E \Bigr{(}\overline{Z}_{n-1}^{\delta}(t) \Bigl{)}$. For $(\tau_2,...,\tau_n)\in \mathcal{S}_{t+\delta_1}^{n-1}(L\delta,\mathfrak{F}^{t+\delta_1})$ write
                \[
                   \Delta_k :=  t_k - t \quad \text{und} \quad  \tau^k:= \ (\tau_2 + \Delta_k ,..., \tau_n + \Delta_k)
                \]
                 for all $k\in\N$.
                By Lemma \ref{lemma eigenschaften stopp zufaellig}.(iv) $\tau_k\in\mathcal{S}_{t+ \Delta_k +\delta_1}^{n-1}(L\delta,\mathfrak{F}^{t+\Delta_k +\delta_1})=\mathcal{S}_{t_k +\delta_1}^{n-1}(L\delta,\mathfrak{F}^{t_k +\delta_1})$ for all $k\in \N$, yielding the inequality
                \[
                    \E \, \sum_{i=2}^n Y(\tau_i)
                     =  \lim_{k\rightarrow \infty} \E \, \sum_{i=2}^n Y(\tau_i^k)
                     \leq  \lim_{k\rightarrow \infty} \E \Bigr{(} Z_{n-1}^{L\delta,\mathfrak{F}^{t_k +\delta_1}}(t_k + \delta_1)\Bigl{)}
                     =  \lim_{k\rightarrow \infty} \E \Bigr{(} \overline{Z}_{n-1}^{\delta}(t_k) \Bigl{)}.
                \]
            Putting pieces together we obtain $ \E \Bigr{(} \overline{Z}_{n-1}^{\delta}(t_k) \Bigl{)}\rightarrow \E \Bigr{(}\overline{Z}_{n-1}^{\delta}(t) \Bigl{)}$.
            \item{${\text{(ii):}}$} For all $\lambda{\geq 0}$ we obtain using Lemma \ref{lemma abz stoppteiten filtration}
                \begin{eqnarray*}
                    1_{\{ \tau=\lambda\}} \, \overline{Z}_{n-1}^{\delta}(\tau)
                    & = &  1_{\{ \tau=\lambda\}} \, \overline{Z}_{n-1}^{\delta}(\lambda)\\
                    & = &  1_{\{ \tau=\lambda\}} \, \E\bigr{(}Z_{n-1}^{L\delta,\mathfrak{F}^{\lambda+\delta_1}}(\lambda+\delta_1)|\mathcal{F}_\lambda \bigl{)}\\
                    & = & \E\bigr{(} 1_{\{ \tau=\lambda\}} \, Z_{n-1}^{L\delta,\mathfrak{F}^{\lambda+\delta_1}}(\lambda+\delta_1)| \mathcal{F}_\lambda\mid_{\{ \tau=\lambda\}} \bigl{)}\\
                    & = &  \E\bigr{(} 1_{\{ \tau=\lambda\}} \, Z_{n-1}^{L\delta,\mathfrak{F}^{\tau+\delta_1}}(\tau+\delta_1)| \mathcal{F}_\tau\mid_{\{ \tau=\lambda\}} \bigl{)}\\
                    & = & 1_{\{ \tau=\lambda\}} \, \E\bigr{(}Z_{n-1}^{L\delta,\mathfrak{F}^{\tau+\delta_1}}(\tau+\delta_1)|\mathcal{F}_\tau \bigl{)}.
                \end{eqnarray*}
                Since $\tau$ has countable range we obtain (\ref{eigenschaft zuf wart bed erw}).
       \end{description}
\end{proof}

Now we come to the main result:
\begin{satz}
    The process $Y_{n}:= Y+ \overline{Z}_{n-1}^{\delta,r}$ is $\mathfrak{F}$-adapted, right continuous and fulfills
            \begin{equation}
                Z_{n}^{\delta,\mathfrak{F}}(\sigma)_{\textnormal{disc}}= \underset{\tau\in \mathcal{S}_\sigma}{\esssup}\; \E(Y_{n}(\tau)|\mathcal{F}_\sigma)
            \end{equation}
              for all $\sigma\in\mathcal{S}$.\\
              If $\delta_1,...,\delta_{n-1}$ are $\mathcal{F}$-stopping times, then furthermore
              \begin{equation}\label{eigenschaft red kl fall zu wart}
                Z_{n}^{\delta,\mathfrak{F}}(\sigma)= \underset{\tau\in \mathcal{S}_\sigma}{\esssup}\; \E(Y_{n}(\tau)|\mathcal{F}_\sigma).
            \end{equation}
\end{satz}

\begin{proof}
    Using (\ref{eigenschaft ind schritt zuf wart disk fall}) and (\ref{eigenschaft zuf wart bed erw}) we obtain    \begin{eqnarray*}
        Z_{n}^{\delta,\mathfrak{F}}(\sigma)_{\text{disc}}
        & = & \underset{\tau\in \mathcal{S}_{\sigma,\text{disc}}}{\esssup} \; \E\Bigr{(} Y(\tau)+ \E\bigr{(}Z_{n-1}^{L\delta,\mathcal{F}^{\tau+\delta_1}}(\tau+\delta_1)|\mathcal{F}_\tau \bigl{)}  \Big{|}\mathcal{F}_\sigma \Bigl{)}\\
        & = & \underset{\tau\in \mathcal{S}_{\sigma,\text{disc}}}{\esssup} \; \E\Bigr{(} Y(\tau)+ \overline{Z}_{n-1}^{\delta,r}(\tau) \Big{|}\mathcal{F}_\sigma \Bigl{)}\\
        & = & \underset{\tau\in \mathcal{S}_{\sigma,\text{disc}}}{\esssup} \; \E( Y_n(\tau) {|}\mathcal{F}_\sigma )\\
        & = & \underset{\tau\in \mathcal{S}_{\sigma}}{\esssup} \; \E ( Y_n(\tau) {|}\mathcal{F}_\sigma ),
    \end{eqnarray*}
    where the last equality holds by approximation. (\ref{eigenschaft red kl fall zu wart}) holds by Proposition \ref{satz zus zuf wart disk nicht disk}.
\end{proof}

By applying the arguments given in the proof of \cite[Proposition 3.2]{Car} to $\overline{Z}_{n-1}^{\delta,r}$, we get

\begin{corollary}\label{folgerung zuf f stoppzeiten}
    If $\delta_1,...,\delta_{n-1}$ are $\mathfrak{F}$-stopping times, then
    \[
       \overline{Z}_{n-1}^{\delta,r}(\tau)=
                \E\bigr{(}Z_{n-1}^{L\delta}(\tau+\delta_1)|\mathcal{F}_\tau \bigl{)}
    \]
    for each $\mathfrak{F}$-stopping time $\tau$.
\end{corollary}

To end this section we remark how the model with random waiting times in discrete time can be found in the model described above. To this end we consider
\[
   \mathcal{F}_t=\mathcal{F}_{\lfloor t \rfloor}, \quad Y(t)= Y(\lfloor t \rfloor) \quad \text{and} \quad \text{range}(\delta_1),..., \text{range}(\delta_{n-1})\subseteq \N_0
\]
for all $t\geq 0$. For each $(\mathcal{F}_k)_{k\in\N_0}$-stopping time $\sigma$ write
\[
            \mathcal{T}_\sigma^n(\delta,\mathfrak{F}):=\Bigr{\{}(\tau_1,...,\tau_n)\in  \mathcal{S}_\sigma^n(\delta,\mathfrak{F});\; \text{range($\tau_1$),..., range($\tau_n$) $\subseteq \N_0$}  \Bigl{\}}.
\]

\begin{lem}\label{lemma zuf wart disk} For each $(\mathcal{F}_k)_{k\in\N_0}$-stopping time $\sigma$:
    \begin{enumerate}[(i)]
        \item
        $\mathcal{T}_\sigma^n(\delta,\mathfrak{F})\subseteq \mathcal{S}_\sigma^n(\delta,\mathfrak{F})_{\textnormal{disc}}$.
        \item
        $\mathcal{F}^{\rho_1,...,\rho_i}_t = \mathcal{F}^{\rho_1,...,\rho_i}_{\lfloor t \rfloor}$ for each $\rho_1,...,\rho_n:\Omega \rightarrow \N_0$, $t\geq 0$, $i{\leq n}$.
       \item
         $\mathcal{F}^{\tau_1+\delta_1,...,\tau_i+\delta_i}_t \subseteq \mathcal{F}^{\lfloor \tau_1 \rfloor +\delta_1,...,\lfloor \tau_i \rfloor +\delta_i}_t$ and $\lfloor \tau_{i} \rfloor$ is an $\mathfrak{F}^{\lfloor \tau_1 \rfloor +\delta_1,...,\lfloor \tau_{i-1} \rfloor +\delta_{i-1}}$-stopping time for each $\tau\in \mathcal{S}_\sigma^n(\delta,\mathfrak{F})$, $t\geq 0$, $i{\leq n}$.
      \item
        $(\lfloor \tau_1 \rfloor,..., \lfloor \tau_n \rfloor)\in \mathcal{T}_\sigma^n(\delta,\mathfrak{F})$ for each $\tau\in \mathcal{S}_\sigma^n(\delta,\mathfrak{F})$.
    \end{enumerate}
\end{lem}

\begin{proof}
The proof is a straightforward exercise.
\end{proof}

So the discrete model may be treated as follows: 
On the one hand by Proposition \ref{satz zus zuf wart disk nicht disk}.(i) and Lemma \ref{lemma zuf wart disk}.(i)
            \[
                Z_n^{\delta,\mathfrak{F}}(\sigma) \geq  Z_n^{\delta,\mathfrak{F}}(\sigma)_{\text{disc}}
                \geq  \underset{\tau\in \mathcal{T}_\sigma^n(\delta,\mathfrak{F})}{\esssup}  \E\Bigr{(}\sum_{i=1}^nY(\tau_i)\Big{|}\mathcal{F}_\sigma\Bigl{)}.
            \]
On the other hand using Lemma \ref{lemma zuf wart disk}.(iv) we have for each $\tau\in \mathcal{S}_\sigma^n(\delta,\mathfrak{F})$
             \begin{eqnarray*}\qquad
                \E\Bigr{(}\sum_{i=1}^nY(\tau_i)\Big{|}\mathcal{F}_\sigma\Bigl{)} = \E\Bigr{(}\sum_{i=1}^nY(\lfloor \tau_i \rfloor)\Big{|}\mathcal{F}_\sigma\Bigl{)} \leq \underset{\rho\in \mathcal{T}_\sigma^n(\delta,\mathfrak{F})}{\esssup} \E\Bigr{(}\sum_{i=1}^nY(\rho_i)\Big{|}\mathcal{F}_\sigma\Bigl{)}.
             \end{eqnarray*}
Therefore we get for each $(\mathcal{F}_k)_{k\in\N_0}$-stopping time $\sigma$
    \begin{equation}\label{eigenschaft disk fall zuf wart}
        Z_n^{\delta,\mathfrak{F}}(\sigma)= Z_n^{\delta,\mathfrak{F}}(\sigma)_{\text{disc}}
        = \underset{\tau\in \mathcal{T}_\sigma^n(\delta,\mathfrak{F})}{\esssup} \E\Bigr{(}\sum_{i=1}^nY(\tau_i)\Big{|}\mathcal{F}_\sigma\Bigl{)}.
    \end{equation}

\section{The Markovian case}\label{markov}
In this section we assume $(X(t),\mathcal{F}_t)_{t\geq 0}$ to be a strong Markov process with state space $E$ and fix a discounting rate $\beta>0$. Furthermore let $h:E \rightarrow [0,\infty)$ be measurable such that the $\mathfrak{F}$-adapted and non-negative process
\begin{equation}\label{eigneschaft defi ausz markov}
    Y = \Bigr{(}e^{-\beta t}h(X(t))\Bigl{)}_{t{\geq 0}}
\end{equation}
is right-continuous, i.e. $\mathbb{P}_x$-right continuous for all $x\in E$ and we assume that
\begin{equation}\label{eigenschaft auz prozess endl}
    \E_x\Bigr{(} \sup_{t{\geq 0}} Y(t) \Bigl{)} < \infty
\end{equation}
holds.\\
In the following we assume the waiting times $\delta_1,...,\delta_{n-1}$ to be $\mathfrak{F}$-stopping times. In particular we have $\mathfrak{F}^{\tau_1+\delta_1,...,\tau_{n-1}+\delta_{n-1}}=\mathfrak{F}$ for all $\tau\in\mathcal{S}_0^n(\delta,\mathfrak{F})$.

\begin{defi}\label{defi Z_n(t,x)}
    For each $x\in E$  let
    \begin{enumerate}[(i)]
        \item $Z_n^\delta(\sigma,x):= \underset{\tau\in \mathcal{S}_\sigma^n(\delta)}{\esssup} \; \E_x\Bigr{(}\underset{i=1}{\overset{n}{\sum}}Y(\tau_i)\Big{|}\mathcal{F}_\sigma\Bigl{)}$ for all $\sigma\in\mathcal{S}$,
        \item $V_n^\delta(x):= Z_n^\delta(0,x)$,
        \item $\overline{Z}_{n-1}^{\delta}(t,x):=  \E_x(Z_{n-1}^{L\delta}(t+\delta_1,x)|\mathcal{F}_t)$ for all $t\geq 0$.
        \item Denote the right-continuous modification of $\overline{Z}_{n-1}^{\delta}(\cdot,x)$ introduced in Lemma \ref{defilem zuf wart supermart} by $\overline{Z}_{n-1}^{\delta,r}(\cdot,x)$.
    \end{enumerate}
\end{defi}

By using the stopping time $\delta_1$ instead of the deterministic $\delta$ in the proofs of \cite[Chapter 4]{Car}, we obtain the following results:

\begin{lem}\label{defilem g_n zuf}
The function $g_{n-1}^\delta:E\rightarrow[0,\infty)$  given by
    \begin{equation}
        g_{n-1}^{\delta}(x):= \E_x\Bigr{(} e^{-\beta \delta_1} V_{n-1}^{L\delta}(X(\delta_1))\Bigl{)}
    \end{equation}
    is measurable and fulfills the following properties:
    \begin{enumerate}[(i)]
        \item
            $g_{n-1}^{\delta}$ is $\beta$-excessive and $C_0$-continuous.
        \item     The process $\Bigr{(} e^{-\beta t} g_{n-1}^{\delta}(X(t))\Bigl{)}_{t\geq 0}$ is right-continuous.
        \item For all $x\in E$ and $\tau\in\mathcal{S}$ it holds that
            \begin{equation}\label{eigenschaft overline Z(dot) zuf}
                e^{-\beta \tau} g_{n-1}^{\delta}(X(\tau))= \overline{Z}_{n-1}^{\delta,r}(\tau,x) \quad \text{$\mathbb{P}_x$-a.s}.
            \end{equation}
    \end{enumerate}
    \qed
\end{lem}

\begin{satz}\label{satz exzessive darstellung von n mal stoppen zuf}
    For the non-negative and $C_0$-continuous function
    \[
        h_{n}^\delta:= h + g_{n-1}^{\delta}
    \]
    it holds that
    \begin{equation}\label{eigenschaft Z_n(sigma,x) zuf}
        Z_{n}^\delta(\sigma,x)=\underset{\tau\in\mathcal{S}_\sigma}{\esssup}\; \E_x\Bigr{(}e^{-\beta\tau} h_{n}^\delta(X(\tau)) \Big{|} \mathcal{F}_\sigma\Bigl{)} \quad \text{$\mathbb{P}_x$-a.s.},
    \end{equation}
    in particular
    \begin{equation}\label{eigenschaft V_n(X(t)) zuf}
        e^{-\beta \sigma} V_{n}^\delta(X(\sigma)) = \underset{\tau\in\mathcal{S}_\sigma}{\esssup}\; \E_x\Bigr{(}e^{-\beta\tau} h_{n}^\delta(X(\tau)) \Big{|} \mathcal{F}_\sigma\Bigl{)} \quad \text{$\mathbb{P}_x$-a.s.}
    \end{equation}
    for all $x\in E$ and $\sigma\in\mathcal{S}$. \qed
\end{satz}

The following lemma shows that the stopping set of the function $h$ given in (\ref{eigneschaft defi ausz markov}) is a subset of the stopping set given by $h_n^\delta$:

\begin{lem}\label{lemma stoppgebiert zuf}
    Let $g:E \rightarrow [0,\infty)$ be measurable such that $\Bigr{(} e^{-\beta t} g(X(t)) \Bigl{)}_{t\geq 0}$ is a right-continuous and non-negative supermartingale with $\sup_{t\geq 0} \E_x\Bigr{(} e^{-\beta t} g(X(t)) \Bigl{)} < \infty $ for all $x\in E$. Then the stopping set given by $h$ is a subset of the stopping set given by $h+g$.
\end{lem}

\begin{proof}
    By the optional sampling theorem for non-negative supermartingales we obtain for each $x$ in the stopping set given by $h$ that
    \begin{eqnarray*} \qquad
        \sup_{\tau \in \mathcal{S}} \E_x\Bigr{(} e^{-\beta\tau} (h+g)(X(\tau)) \Bigl{)}
        & \leq  & \sup_{\tau \in \mathcal{S}} \E_x\Bigr{(} e^{-\beta\tau} h(X(\tau)) \Bigl{)} + \sup_{\tau \in \mathcal{S}} \E_x\Bigr{(} e^{-\beta\tau} g(X(\tau)) \Bigl{)}\\
        & = & h(x) + \E_x\Bigr{(} e^{-\beta 0} g(X(0)) \Bigl{)}\\
        & = & h(x) + g(x).
    \end{eqnarray*}
\end{proof}

\begin{corollary}\label{folgerung zuf stoppgebiet}
    $\{ x\in E;\; V_1(x)=h(x)\} \subseteq \{ x\in E;\; V_n^\delta(x)=h_n(x)\}$.
\end{corollary}

\begin{proof}
    By  (\ref{eigenschaft overline Z(dot) zuf}) the process
    $\Bigr{(} e^{-\beta t} g_{n-1}^\delta(X(t)) \Bigl{)}_{t\geq 0}$ is a right-continuous and non-negative supermartingale with $\sup_{t\geq 0} \E_x\Bigr{(} e^{-\beta t} g_{n-1}^\delta(X(t)) \Bigl{)} < \infty $
    for all $x\in E$. Keeping (\ref{eigenschaft V_n(X(t)) zuf}) and Lemma \ref{lemma stoppgebiert zuf} in mind this proves the result.
\end{proof}

\begin{anm}
The case of discrete time Markov processes may be treated the same way.
\end{anm}

\subsection{Example: Multiple house-selling problem }\label{subsection:house}
As an application of the theory developed before we consider a multiple-exercise variant of the classical house-selling problem. The situation is the following: We would like to sell $n$ identical houses and assume that at each time point one offer comes in for one of the houses. By $X_k$ we denote the amount of the offer on day $k$, $k=0,1,...$. 
For each offer we have to decide whether to accept it or not. If we decide to accept the offer, then closing the contract lasts a random time $\delta\geq 1$, that is assumed to be independent of the offers. During that refraction time we cannot deal with new offers. When should we accept an offer if we are not able to recall and accept a past offer?\\
For the single exercise case the problem was treated in \cite{CRS}. We model the situation in the following way: Let $X_0,X_1,...$ be non-negative iid random variables with $\E(X_0^2)<\infty$. The last assertion guarantees that condition \eqref{eigenschaft supremumsprozess int} is fulfilled. Then we consider the discrete time stochastic process $Y(k)=\alpha^kX_k,k=0,1,2,...$, where $\alpha\in(0,1)$ is a fixed discounting factor. Furthermore, we consider i.i.d. refraction times $\delta_1,...,\delta_{n-1}$ that are assumed to be $>0$ and independent of $(X_0,X_1,...)$. Now, the problem is to maximize the expectation
\begin{equation}\label{eq:max_house}
\E\Bigr{(}\sum_{i=1}^nY(\tau_i)\Bigl{)}=\E\Bigr{(}\sum_{i=1}^n\alpha^{\tau_i}X_{\tau_i}\Bigl{)}
\end{equation}
over all $(\tau_1,...,\tau_n)\in  \mathcal{S}_0^n(\delta,\mathfrak{F})$, where $\mathfrak{F}$ denotes the filtration generated by $X_0,X_1,...$. Now, we enrich the filtration $\mathfrak{F}$ by $\delta_1,...,\delta_{n-1}$, i.e. we define $\overline{\mathfrak{F}}$ by $\overline{\mathcal{F}}_k={\mathcal{F}}_k\vee \sigma(\delta_1,...,\delta_{n-1})$. Trivially, the stochastic process $X_0,X_1,...$ is still Markovian w.r.t. $\overline{\mathfrak{F}}$. Furthermore, $\mathcal{S}_0^n(\delta,\mathfrak{F})\subseteq \mathcal{S}_0^n(\delta,\overline{\mathfrak{F}})$. We first maximize \eqref{eq:max_house} over the set $\mathcal{S}_0^n(\delta,\overline{\mathfrak{F}})$ and then see that the maximizer is indeed an element of $\mathcal{S}_0^n(\delta,\mathfrak{F})$. Using Theorem \ref{satz exzessive darstellung von n mal stoppen zuf} we see that 
\[Z_{n}^\delta(\sigma,x)=\underset{\tau\in\mathcal{S}_\sigma(\overline{\mathfrak{F}})}{\esssup}\; \E_x\Bigr{(} \alpha^\tau (X_\tau+g_{n-1}^{\delta}(X_\tau))\Big{|} \mathcal{F}_\sigma\Bigl{)},\]
where $g_{n-1}^\delta(x)=\E_x(\alpha^\delta V_{n-1}(X_{\delta})), \delta=\delta_{1}$. Conditioning on $\delta$ we see that $X_\delta$ has the same distribution as $X_1$. Therefore, we obtain that 
\[g_{n-1}^\delta(x)=\E(V_{n-1}(X_1))\E(\alpha^\delta)=:d_{n-1}\]
is independent of $x$. Moreover, 
\[Z_{n}^\delta(\sigma,x)=\underset{\tau\in\mathcal{S}_\sigma(\overline{\mathfrak{F}})}{\esssup}\; \E_x\Bigr{(} \alpha^\tau Z_{\tau}^{n-1}\Big{|} \mathcal{F}_\sigma\Bigl{)},\]
where ${Z_k^{n-1}}=X_k+d_{n-1}$. Consequently, we have reduced the multiple optimal stopping problem \eqref{eq:max_house} to the ordinary stopping problem for the case $n=1$ with random variables with adjusted distributions. The $d_{n-1}$ are computed successively and at each step the distribution is shifted by this constant quantity. Therefore, to solve the general problem we only have to consider the problem of maximizing
\[\E_x(\alpha^\tau X_\tau)\]
over all stopping times $\tau$ for a sequence of random variables $X_i,i=0,1,2,...$, see \cite{CRS}. It is well-known how to do this. We repreat the simple argument for completeness. We denote the value function by $v(x)$. Using the Bellman-principle we see that 
\begin{equation}\label{eq:bellman}
v(x)=\max\{x,\alpha \E_x( v(X_1))\}
\end{equation}
 for all $x$. Since $\alpha\E_x( v(X_1))$ is independent of $x$, we see that there exists $x^*$ such that $v(x)=x$ for $x\geq x^*$, and $v(x)=x^*$ for $x\leq x^*$, i.e. $v(x)=(x-x^*)^++x^*$. Equation \eqref{eq:bellman} now yields for $x=x^*$
\[x^*=\alpha\E((X_1-x^*)^++x^*)\]
that is
\[\frac{1-\alpha}{\alpha}=\E\left(\left(\frac{X_1}{x^*}-1\right)^+\right).\]
Since the right hand side is decreasing in $x^*$, we see that this determines $x^*$ uniquely. The optimal stopping time is now given by 
\[\inf\{k:X_k\geq x^*\}.\]

\section{A second model for random refraction times}\label{sec:second_model}
%
For a motivation consider employee options:
As a variable component of the salary the holder of the option has the right to exercise an option on the companies share price $n$ times during a given time period. But the total amount of variable compensations must not restrict the institution's ability to maintain an adequate capital base. Therefore the holder has to wait between two exercises, e.g., until the liquid assets of the company excess a certain level. Of course the waiting time is not deterministic, but random. But since the waiting time directly depends on the foregoing exercise time this situation is not included in the previous model. Motivated by this example we consider the following situation:

Let $B_1,...,B_{n-1}$ be Borel sets and let $X=(X(t))_{t\geq 0}$ be a further $\mathfrak{F}$-adapted and right-continuous process such that for all $i<n$ and $s\geq 0$
\begin{equation}\label{eq:finite}
    \bigcup_{t>s}\{ X(t)\notin B_i\}\quad \text{is a $\mathbb{P}$ null set}
\end{equation}
and
\begin{equation}\label{eigenschaft stoppzeit zu wart}
   \rho_s^{B_i}:= \inf\{t>s;\; X(t)\in B_i\} \quad \text{is an $\mathfrak{F}$-stopping time;}
\end{equation}
so, refraction times are assumed $<\infty$ a.s.
Then we consider the problem of determining
    \begin{equation}\label{eigenschaft defi stoppproblem mit zuf wart 2}
        Z_n(\sigma)= Z_n^{B}(\sigma):= \underset{\tau\in \mathcal{S}_\sigma^n(B)}{\esssup}\; \E\Bigr{(}\sum_{i=1}^nY(\tau_i)\Big{|}\mathcal{F}_\sigma\Bigl{)},
    \end{equation}
    where
    \[
        \mathcal{S}_\sigma^n(B):=\{\tau\in\mathcal{S}^n;\; \sigma \leq \tau_1, \; \rho_{\tau_i}^{B_i} \leq  \tau_{i+1} \, \text{for all $i< n$}\}.
    \]

%
In the employee option problem described above $X$ is modeled as the liquid assets of the company; of course this process is not independent of the gain process $Y$. Furthermore $B_1,...,B_{n-1}$ are bounded intervals of the form $[b_i,\infty)$.\\

As in \eqref{eigenschaft disk fall zuf wart} we have:
\begin{lem}\label{lemma zuf wart mod 2 disk}
    \begin{enumerate}[(i)]
    \item
    We use the notation
    \[
        \mathcal{S}_\sigma^n(B)_{\textnormal{disc}}:= \{\tau \in\mathcal{S}_\sigma^n(B); \textnormal{range}(\tau_1),...,\textnormal{range}(\tau_n) \, \textnormal{countable}\}.
    \]
     Then  for each stopping time $\sigma$
    \begin{equation}
         Z_n^{B}(\sigma)= \underset{\tau\in \mathcal{S}_\sigma^n(B)_{\textnormal{disc}}}{\esssup} \E\Bigr{(}\sum_{i=1}^nY(\tau_i)\Big{|}\mathcal{F}_\sigma\Bigl{)}.
    \end{equation}
    \item
     The set
     $\{\E\Bigr{(}\underset{i=1}{\overset{n}{\sum}}Y(\tau_i)\Big{|} \mathcal{F}_\sigma\Bigl{)};\; \tau\in\mathcal{S}_\sigma^{n}(B)\}$ is directed upwards.
     \end{enumerate}
\end{lem}

\begin{proof} $(i)$ is a straightforward exercise and $(ii)$ is similar to (but easier than) the proof of Lemma \ref{lemma oben gerich} in \ref{proof lemma oben gerich}.
\end{proof}
%
%
Now we can adapt the proof of Theorem \ref{sazu zuf wart ind schr} and obtain the first step of the reduction principle:
\begin{satz}\label{satz red schritt zuf wart 2}
    For all stopping times $\sigma$, using the notation $LB:=(B_2,...,B_{n-1})$, it holds that
    \begin{eqnarray}
         Z_{n}^{B}(\sigma)
         & = & \underset{\tau\in \mathcal{S}_{\sigma, \textnormal{disc}}}{\esssup} \; \E\Bigr{(} Y(\tau)+ \E\bigr{(}Z_{n-1}^{LB}(\rho_\tau^{B_1})|\mathcal{F}_\tau \bigl{)}
                \Big{|}\mathcal{F}_\sigma \Bigl{)}\\
         & = &   \underset{\tau\in \mathcal{S}_\sigma}{\esssup} \;  \E\Bigr{(} Y(\tau)+ \E\bigr{(}Z_{n-1}^{LB}(\rho_\tau^{B_1})|\mathcal{F}_\tau \bigl{)}
                \Big{|}\mathcal{F}_\sigma \Bigl{)}.
    \end{eqnarray}
    \qed
\end{satz}

In contrast to the model of Section \ref{Abschnitt zuf wart} we obtain the second claim of the following Lemma immediately from the definition of the stopping problem (\ref{eigenschaft defi stoppproblem mit zuf wart 2}):

\begin{lem}\label{defilem zuf wart supermart 2}
    The process defined by
    \begin{equation}\label{eigenschaft defi zuf wart quer}
        \overline{Z}_{n-1}^{B}(t):=
        \E\bigr{(}Z_{n-1}^{LB}(\rho_t^{B_1})|\mathcal{F}_t \bigl{)}
    \end{equation}
    is a supermartingal $\overline{Z}_{n-1}^{B}$ with the following properties:
    \begin{enumerate}[(i)]
        \item
            $\overline{Z}_{n-1}^{B}$ has a right-continuous modification
            $\overline{Z}_{n-1}^{B,r}$.
        \item For all $\tau\in\mathcal{S}$ with countable range it holds that
            \begin{equation}\label{eigenschaft zuf wart bed erw 2}
                \overline{Z}_{n-1}^{B,r}(\tau)= \overline{Z}_{n-1}^{B}(\tau)=
                \E\bigr{(}Z_{n-1}^{LB}(\rho_\tau^{B_1})|\mathcal{F}_\tau \bigl{)}.
            \end{equation}
    \end{enumerate}
\end{lem}

\begin{proof}
For all $t\geq s\geq 0$ we have $s \leq \rho_s^{B_1} \leq \rho_t^{B_1}$ and hence
        \begin{eqnarray*}
        \E(\overline{Z}_{n-1}^{B}(t) |\mathcal{F}_s)
          =  \E\Bigr{(} \E(Z_{n-1}^{LB}(\rho_t^{B_1})|\mathcal{F}_{\rho_s^{B_1}}) \Big{|}\mathcal{F}_s \Bigl{)}
         \overset{(*)}{\leq}  \E(Z_{n-1}^{LB}(\rho_s^{B_1})|\mathcal{F}_s)
        = \overline{Z}_{n-1}^{B}(s).
    \end{eqnarray*}
    Here (*) is valid since
    \begin{eqnarray} \label{eigenschaft rechnung 2}
         \E(Z_{n-1}^{LB}(\rho_t^{B_1})|\mathcal{F}_{\rho_s^{B_1}})
         & = &  \E\Bigr{[}\underset{\tau\in \mathcal{S}_{\rho_t^{B_1}}^{n-1}(LB)}{\esssup} \; \E\Bigr{(}\sum_{i=1}^{n-1}Y(\tau_i)\Big{|}\mathcal{F}_{\rho_t^{B_1}}\Bigl{)}\Big{|}\mathcal{F}_{\rho_s^{B_1}}\Bigl{]}\\\nonumber
         & \leq & \E\Bigr{[}\underset{\tau\in \mathcal{S}_{\rho_s^{B_1}}^{n-1}(LB)}{\esssup} \; \E\Bigr{(}\sum_{i=1}^{n-1}Y(\tau_i)\Big{|}\mathcal{F}_{\rho_t^{B_1}}\Bigl{)}\Big{|}\mathcal{F}_{\rho_s^{B_1}}\Bigl{]}\\\nonumber
         & \overset{\ref{lemma zuf wart mod 2 disk}.\textnormal{(ii)}}{=} & \underset{\tau\in \mathcal{S}_{\rho_s^{B_1}}^{n-1}(LB)}{\esssup} \; \E\Bigr{(}\sum_{i=1}^{n-1}Y(\tau_i)\Big{|}\mathcal{F}_{\rho_s^{B_1}}\Bigl{)}\\ \nonumber
         & = & Z_{n-1}^{LB}(\rho_s^{B_1}).
    \end{eqnarray}
    Now Property (i) holds by \cite[Theorem 3.13]{KS} and
    (ii) holds by the definition of $\overline{Z}_{n-1}^{B}$.
\end{proof}

Therefore we obtain

\begin{satz}\label{satz zuf wart alt modell hs}
     $ Y_n:=Y_{n}^B:= Y+ \overline{Z}_{n-1}^{B,r}$ is an $\mathfrak{F}$-adapted and right-continuous stochastic process and fulfills
    \begin{equation}\label{eigenschaft neues stoppp zuf wart 2}
        Z_{n}^{B}(\sigma)= \underset{\tau\in \mathcal{S}_\sigma}{\esssup}\; \E(Y_{n}(\tau)|\mathcal{F}_\sigma)
    \end{equation}
    for all $\sigma\in\mathcal{S}$.
\end{satz}

\begin{proof}
    Applying Theorem \ref{satz red schritt zuf wart 2} and Lemma \ref{defilem zuf wart supermart 2}.(ii) and using the right-continuity of $Y_n$ yields
    \begin{eqnarray*} \qquad\qquad\qquad\qquad\qquad\,
        Z_{n}^{B}(\sigma)
         & = & \underset{\tau\in \mathcal{S}_{\sigma, \textnormal{disc}}}{\esssup} \; \E\Bigr{(} Y(\tau)+ \E\bigr{(}Z_{n-1}^{LB}(\rho_\tau^{B_1})|\mathcal{F}_\tau \bigl{)}
                \Big{|}\mathcal{F}_\sigma \Bigl{)}\\
         & = & \underset{\tau\in \mathcal{S}_{\sigma, \textnormal{disc}}}{\esssup} \; \E\Bigr{(} Y_n(\tau)|\mathcal{F}_\tau \bigl{)}
                \Big{|}\mathcal{F}_\sigma \Bigl{)}\\
         & = &    \underset{\tau\in \mathcal{S}_\sigma}{\esssup} \; \E\Bigr{(} Y_n(\tau)|\mathcal{F}_\tau \bigl{)}
                \Big{|}\mathcal{F}_\sigma \Bigl{)}.
    \end{eqnarray*}
\end{proof}

\section{Example:\ Perpetual put option}\label{sec:example}
As an example we consider a perpetual put option with multiple stopping in a Black-Scholes market with waiting times depending on the asset price. To be more precise let  $W$ be a Brownian motion and let $\mathfrak{F}$ be the right-continuous filtration generated by $W$. Furthermore let $K,\sigma{> 0}$ , $E=(0,\infty)$, $h(x)= (K-x)^+$ and the asset price process $A$ be given by
\[
  A(t)= x \, \exp\Bigr{(} \sigma W_t + ( \beta - \frac{\sigma^2}{2})t \Bigl{)}
\]
under $\mathbb{P}_x$. The reward process $Y$ is  given by
\[Y(t)=e^{-\beta t}h(A_t).\]
 For deterministic waiting times it can be shown that there exist $x_1^*\leq x_2^* \leq ...\leq x_n^*$ such that the threshold times for $A$ over $x_i$ is the optimal time for the $i$-th exercise, see \cite[Section 6.2]{Car}. But it seems very hard to determine $x_2^*,x_3^*,...$ explicitly in non-trivial examples. However we can find an explicit solution when dealing with some random waiting times in our second model:\\
We set $X=A$ and assume that between each two exercises we always have to wait until the process reaches a level $z_0\geq K$, i.e. $B_i=[z_0,\infty)$ for all $i$. For the refraction times to be finite -- see \eqref{eq:finite} -- we assume that $\beta\geq \sigma^2/2$.  Since the problem has a Markovian structure it is reasonable to use this as  in Section \ref{markov}. The results and notations given there can immediately be taken over to the second model and we use them in the following.\\
It is well-known that with only one stopping opportunity (i.e. $n=1$) it is optimal to stop when the process reaches
\[x_1^*:= \frac{K}{1+ \sigma^2/ 2\beta},\]
see e.g. \cite{BL}. Then for $x\leq z_0$ the function $g_1^{(\delta_1)}$ is given by
\begin{align}\label{g_1}
g_1^{(\delta_1)}(x)&:=\E_x(e^{-\beta\delta_1}V_1(A(\delta_1)))=V_1(z_0)\E_x(e^{-\beta \delta_1})=c_1x,
\end{align}
and for $x> z_0$
 \begin{align*}
g_1^{(\delta_1)}(x)&:=V_1(x),
\end{align*}
where
\[c_1=c_1(z_0)=\frac{V_1(z_0)}{z_0}=\frac{K-x_1^*}{z_0}\left(\frac{z_0}{x_1^*}\right)^\gamma, \gamma=-\frac{2r}{\sigma^2},\] see e.g. \cite{Borodin} for the last equality in \eqref{g_1}. The explicit representation of $g_1^{(\delta_1)}$ is the main reason why an explicit solution is possible in this example. Now note that
\begin{align*}
\frac{V_1(K)}{K}&=\frac{(K-x_1^*)\left(\frac{K}{x_1^*}\right)^\gamma}{K}=\left(1-\frac{1}{1+\frac{\sigma^2}{2\beta}}\right)\left(1+\frac{\sigma^2}{2\beta}\right)^\gamma\\
&=\frac{\sigma^2}{2\beta}\left(1+\frac{\sigma^2}{2\beta}\right)^{\gamma-1}< 1.
\end{align*}
Therefore we obtain $c_1< 1$ since $z_0\geq K$.\\
We concentrate on the case $n=2$. In this situation the problem to be solved is given by
\begin{align*}
V_2^{(\delta_1)}(x)&=\sup_\tau \E_x(e^{-\beta\tau}((K-A_\tau)^++g_{1}^{(\delta_1)}(A_{\tau})))\\
&=\sup_\tau \E_x(e^{-\beta\tau}h(A_\tau)),
\end{align*}
where the reward function $h$ is a continuous functions given by
\[h(x)=
\begin{cases}
K-(1-c_1)x&,\,x\leq K,\\
c_1x&,\,K<x\leq z_0,\\
V_1(x)=c_1z_0^{1-\gamma}x^\gamma&,\,x>z_0.
\end{cases}
\]
Obviously it holds that
\begin{equation}\label{eq:maj}
h(x) \geq (K-(1-c_1)x)^+\mbox{ for all $x$.}
\end{equation}
 On the other hand,
\begin{align*}
\sup_\tau\E_x(e^{-\beta\tau}(K-(1-c_1)A_{\tau})^+)
&=\sup_\tau\E_{(1-c_1)x}(e^{-\beta\tau}(K-A_{\tau})^+)\\
&=V_1((1-c_1)x).
\end{align*}
This shows that $x\mapsto V_1((1-c_1)x)$ is $\beta$-excessive and a majorant of $h$. Therefore, keeping \eqref{eq:maj} in mind, we obtain that
\[V_2^{(\delta_1)}(x)=V_1((1-c_1)x).\]
We see that for $n=2$ it is optimal to stop the first time when the process started in $(1-c_1)x$ reaches $x_1^*$, i.e. if the process started in $x$ reaches
\[x_2^*=\frac{x_1^*}{1-c_1}= \frac{K}{(1-c_1)(1+ \sigma^2/ 2\beta)}.\]
 Now the optimal solution for $n>2$ can be found the same way using induction.

\section*{Acknowledgements}
This paper was partly written during the first authors stay at \AA bo Akademi in the project \textit{Applied Markov processes -- fluid queues, optimal stopping and population dynamics, Project number: 127719  (Academy of Finland)}.  He would like to express his gratitude for the hospitality and support.

  \begin{appendix}
  \section{Appendix: Some proofs}
\subsection{Proof of Proposition \ref{satz eigenschaften von neuer filtration}}
\label{proof satz eigenschaften von neuer filtration}
$(i)$ and $(ii)$ are immediate from the definition. \\
For $(iii)$ note that
            \begin{eqnarray*}
                 \mathcal{G}_t^{\sigma+\rho}\mid_{ \{\sigma=\tau\}}
            &  = & \sigma\Bigr{(} \{ \{\sigma + \rho \leq s\};\; s\leq t\}\mid_{ \{\sigma=\tau\}} \Bigl{)}\\
             & = & \sigma\Bigr{(} \{ \{\tau + \rho \leq s\};\; s\leq t\}\mid_{ \{\sigma=\tau\}} \Bigl{)}\\
            & = & \mathcal{G}_t^{\tau+\rho}\mid_{ \{\sigma=\tau\}},
            \end{eqnarray*}
            hence
            \begin{eqnarray*}
              \sigma(\mathcal{A}_{t},\mathcal{G}^{\sigma+\rho}_{t})\mid_{ \{\sigma=\tau\}}
                & = & \sigma(\mathcal{A}_{t},\mathcal{G}^{\tau+\rho}_{t})\mid_{ \{\sigma=\tau\}}.
            \end{eqnarray*}
To prove $(iv)$ write $\text{range}(\sigma) \cup \text{range}(\tau)= \{\lambda_i;\; i\in I\}$, $I$ countable. Then it holds
                   \begin{eqnarray*}
                        \{ \tau+\rho \leq t\}
                        & = & \quad \;\; \bigcup_{i,j\in I} \{\sigma=\lambda_i\}\cap \{ \tau=\lambda_j\}\cap \{\lambda_j+\rho \leq t\}\\
                        & = &  \;\; \bigcup_{i,j\in I,\; \lambda_j\leq t}  \{\sigma=\lambda_i\}\cap \{ \tau=\lambda_j\}\cap \{\lambda_i+\rho \leq t +\lambda_i - \lambda_j\}\\
                        & \overset{\sigma\leq \tau}{=} & \bigcup_{i,j\in I,\, \lambda_i\leq \lambda_j\leq t} \{\sigma=\lambda_i\}\cap \{ \tau=\lambda_j\}\cap \{\sigma+\rho \leq t +\lambda_i - \lambda_j\}\\
                        & \in & \bigcup_{i,j\in I,\, \lambda_i\leq \lambda_j\leq t} \sigma\Bigr{(}\mathcal{A}_{\lambda_i}, \mathcal{A}_{\lambda_j}, \mathcal{A}_{t+\lambda_i - \lambda_j}^{\sigma+\rho}\Bigl{)}  \quad \;\; \subseteq   \quad \;\; \mathcal{A}_{t}^{\sigma+\rho},
                    \end{eqnarray*}
                  therefore
                    \begin{eqnarray*}\qquad\quad\;\;\;
                        \sigma(\mathcal{A}_t,\mathcal{G}^{\tau+\rho}_t) =  \sigma(\mathcal{A}_t,\{\{\tau+\rho \leq s\}, s\leq t\})\subseteq  \sigma(\mathcal{A}_t,\mathcal{A}_{t}^{\sigma+\rho})= \mathcal{A}_{t}^{\sigma+\rho}.
                    \end{eqnarray*}

\subsection{Proof of Lemma \ref{lemma eigenschaften stopp zufaellig}}\label{proof lemma eigenschaften stopp zufaellig}
    \begin{description}
        \item{${\text{(i):}}$}\; This is immediate by definition.
         \item{${\text{(ii):}}$}\; Write
         \[
            (\rho_2,...,\rho_n):=  \Bigr{(}  \tau_2 1_{\{\tau=\lambda\}} +  (\lambda+ \delta_1) 1_{\{\tau\neq\lambda\}},..., \tau_n 1_{\{\tau=\lambda\}} + (\lambda+\delta_1+...+\delta_{n-1}) 1_{\{\tau\neq\lambda\}}  \Bigl{)}.
         \]
         For all $t< \lambda$ and $i\in\{2,...,n\}$ we have
         \[
            \{\rho_i\leq t\}=\emptyset.
         \]
         Using Proposition \ref{satz eigenschaften von neuer filtration}.(iii) recursively we obtain for $i\geq 3$
         \[
            \mathcal{F}^{\tau+\delta_1,\tau_2+\delta_2,...,\tau_{i-1}+\delta_{i-1}}_t\mid_{ \{ \tau=\lambda\}}= \mathcal{F}^{\lambda+\delta_1,\tau_2+\delta_2,...,\tau_{i-1}+\delta_{i-1}}_t\mid_{ \{ \tau=\lambda\}}
         \]
          and using $\{\tau=\lambda\}\subseteq \underset{j=2}{\overset{n}{\bigcap}} \{\tau_j=\rho_j \}$ we furthermore get
         \[
            \mathcal{F}^{\lambda+\delta_1,\tau_2+\delta_2,...,\tau_{i-1}+\delta_{i-1}}_t\mid_{ \{ \tau=\lambda\}}= \mathcal{F}^{\lambda+\delta_1,\rho_2+\delta_2,...,\rho_{i-1}+\delta_{i-1}}_t\mid_{ \{ \tau=\lambda\}}.
         \]
         Therefore
         \begin{eqnarray*}
            \{\rho_i\leq t\}
            & = &  \{\tau_i\, 1_{\{\tau=\lambda\}}\, + \, (\lambda+ \delta_1+...+\delta_{i-1})\, 1_{\{\tau\neq\lambda\}} \leq t\}\\
            & = & \bigr{(}\{ \tau=\lambda\}\cap \{\tau_i\leq t\} \bigl{)} \; \cup \;  \bigr{(}\{ \tau\neq\lambda\}\cap \{ \lambda+ \delta_1+...+\delta_{i-1}\leq t\} \bigl{)}\\
            & = & \bigr{(}\{ \tau=\lambda\}\cap \{\tau_i\leq t\} \bigl{)} \; \cup \;  \bigr{(}\{ \tau\neq\lambda\}\cap \{ \rho_{i-1} + \delta_{i-1}\leq t\} \bigl{)}\\
            & \in & \sigma\Bigr{(} \mathcal{F}^{\tau+\delta_1,\tau_2+\delta_2,...,\tau_{i-1}+\delta_{i-1}}_t\mid_{\{ \tau=\lambda\}},\,  \mathcal{F}_t^{\rho_{i-1}+\delta_{i-1}}\mid_{\{  \tau\neq\lambda\}} \Bigl{)} \\
             & = & \sigma\Bigr{(} \mathcal{F}^{\lambda+\delta_1,\rho_2+\delta_2,...,\rho_{i-1}+\delta_{i-1}}_t\mid_{\{ \tau=\lambda\}} ,\, \mathcal{F}_t^{\rho_{i-1}+\delta_{i-1}}\mid_{\{  \tau\neq\lambda\}}\Bigl{)} \\
            & \subseteq & \sigma\Bigr{(} \mathcal{F}_\lambda, \mathcal{F}^{\lambda+\delta_1,\rho_2+\delta_2,...,\rho_{i-1}+\delta_{i-1}}_t, \mathcal{F}_t^{\rho_{i-1}+\delta_{i-1}} \Bigl{)}\\
            & = & \mathcal{F}^{\lambda+\delta_1,\rho_2+\delta_2,...,\rho_{i-1}+\delta_{i-1}}_t.
         \end{eqnarray*}
            Hence $\rho_i$ is an $\mathfrak{F}^{\lambda+\delta_1,\rho_2+\delta_2,...,\rho_{i-1}+\delta_{i-1}}$-stopping time with $\rho_{i-1}+\delta_{i-1}\leq \rho_i$ for $i\geq 3$. A similar argument also applies for $i=2$.
        \item{${\text{(iii):}}$}\; This is immediate by applying (ii) to
         \[\Bigr{(}  \tau_2 1_{\{\tau=\lambda\}} +  (\tau+ \delta_1) 1_{\{\tau\neq\lambda\}},..., \tau_n 1_{\{\tau=\lambda\}} + (\tau+\delta_1+...+\delta_{n-1}) 1_{\{\tau\neq\lambda\}}  \Bigl{)}.\]
        \item{${\text{(iv):}}$}\; We first prove that for $i\geq 2$
          \begin{equation}\label{eigenschaft lemma fil}
                \mathcal{F}^{\tau+ \delta_1,\tau_2+\delta_2,...,\tau_i+\delta_i}_{t-c}\subseteq \mathcal{F}^{\tau+ c+ \delta_1,\tau_2+c +\delta_2,...,\tau_i+ c+ \delta_i}_{t}.
          \end{equation}
         The case $i=2$ holds by \ref{satz eigenschaften von neuer filtration}.(ii). By induction we obtain again using Proposition \ref{satz eigenschaften von neuer filtration}.(ii)
         \begin{eqnarray*}
          \sigma\Bigr{(} \mathcal{F}^{\tau+ \delta_1,\tau_2+\delta_2,...,\tau_i+\delta_i}_{t-c}, \mathcal{G}^{\tau_{i+1}+\delta_{i+1}}_{t-c} \Bigl{)}& \subseteq & \sigma\Bigr{(} \mathcal{F}^{\tau+c+ \delta_1,\tau_2+c+\delta_2,...,\tau_i+c+\delta_i}_{t}, \mathcal{G}^{\tau_{i+1}+\delta_{i+1}}_{t-c} \Bigl{)}\\ & = & \sigma\Bigr{(} \mathcal{F}^{\tau+c+ \delta_1,\tau_2+c+\delta_2,...,\tau_i+c+\delta_i}_{t}, \mathcal{G}^{\tau_{i+1}+c+\delta_{i+1}}_{t} \Bigl{)},
         \end{eqnarray*}
         hence
         \[
                \mathcal{F}^{\tau+ \delta_1,\tau_2+\delta_2,...,\tau_{i+1}+\delta_{i+1}}_{t-c}\subseteq \mathcal{F}^{\tau+ c+ \delta_1,\tau_2+c +\delta_2,...,\tau_{i+1}+ c+ \delta_{i+1}}_{t}.
         \]
         Therefore we obtain
            \[
                \{ \tau_i + c \leq t \} =  \{ \tau_i  \leq t -c \} \in \mathcal{F}^{\tau+ \delta_1,\tau_2+\delta_2,...,\tau_{i-1}+\delta_{i-1}}_{t-c}\subseteq \mathcal{F}^{\tau+ c+ \delta_1,\tau_2+c +\delta_2,...,\tau_{i-1}+ c+ \delta_{i-1}}_{t},
            \]
            i.e. $\tau_i+c$ is an $\mathfrak{F}^{\tau+ c+ \delta_1,\tau_2+c +\delta_2,...,\tau_{i-1}+ c+ \delta_{i-1}}$-stopping time.

   \end{description}
  \end{appendix}

 \subsection{Proof of Lemma \ref{lemma oben gerich}} \label{proof lemma oben gerich}
   For $\tau,\rho\in\mathcal{S}_\sigma^{n}(\delta)$ let
 \[
    A:=\Bigr{\{ }\E\Bigr{(}\sum_{i=1}^nY(\tau_i)\Big{|}\mathcal{F}_\sigma\Bigl{)} \; \geq \; \E\Bigr{(}\sum_{i=1}^nY(\rho_i)\Big{|}\mathcal{F}_\sigma\Bigl{)} \Bigr{\} }
 \]
 and
 \[
    \nu_i:= 1_A\tau_i\, + \, 1_{A^c}\, \rho_i \mbox{ for }i=1,...,n.
 \]
Then $\nu_1$ is an $\mathfrak{F}$-stopping time with $\sigma\leq \nu_1$ and $\tau_{i-1}+\delta_{i-1}\leq \tau_i$ for all $i\in\{2,...,n\}$. To prove that  $\nu_i$ is an $\mathfrak{F}^{\tau_1+\delta_1,...,\tau_{i-1}+\delta_{i-1}}$-stopping time for $i=2,...,n$ we first prove that
 \begin{equation}\label{eigenschaft ind vor oben ger}
    \mathcal{F}_t^{\tau_1+\delta_1,...,\tau_i+\delta_i}\mid_{ A\cap \{\sigma\leq t\}}\; \subseteq  \mathcal{F}_t^{\nu_1+\delta_1,...,\nu_i+\delta_i} \quad \text{for all $t\geq 0$}.
 \end{equation}
For all $s\leq t$ it holds that  $\mathcal{F}_t\mid_{ A\cap \{\sigma\leq t\}}\; \subseteq  \mathcal{F}_t^{\nu_1+\delta_1}$ and
    \begin{eqnarray} \label{eigenschaft rech im ia}
        & & A\cap \{\sigma\leq t\}\cap \{ \tau_1+\delta_1 \leq s  \}\\ \nonumber
        & = & A\cap \{\sigma\leq t\}\cap A\cap \{ \tau_1+\delta_1 \leq s  \}\\ \nonumber
        &\in & \sigma \Bigr{(} \mathcal{F}_t, \, \{ A\cap \{ \tau_1+\delta_1 \leq r \} ;\; r\leq t  \} \Bigl{)}\\ \nonumber
        & = & \sigma \Bigr{(} \mathcal{F}_t,\, \{ A\cap\{ \sigma \leq r\};\; r\leq t \} ,\, \{ A\cap \{ \tau_1+\delta_1 \leq r \} ;\; r\leq t  \} \Bigl{)}\\ \nonumber
        & \subseteq & \sigma \Bigr{(} \mathcal{F}_t,\, \{ A\cap\{ \sigma \leq r\};\; r\leq t \} ,\, \\ \nonumber
        & & \qquad \quad \{ \Bigr{(}A\cap \{ \tau_1+\delta_1 \leq r \} \Bigl{)} \cup  \Bigr{(}A^c\cap \{ \rho_1+\delta_1 \leq r \}\Bigl{)} ;\; r\leq t  \} \Bigl{)}\\ \nonumber
        & = & \sigma \Bigr{(} \mathcal{F}_t,\, \{ A\cap\{ \sigma \leq r\};\; r\leq t \} ,\, \{ \{ \nu_1+\delta_1 \leq r\} ;\; r\leq t  \} \Bigl{)}\\ \nonumber
        & = & \sigma \Bigr{(} \mathcal{F}_t,\, \{ \{ \nu_1+\delta_1 \leq r\} ;\; r\leq t  \} \Bigl{)}\\ \nonumber
        & = & \sigma\Bigr{(} \mathcal{F}_t,\, \mathcal{G}_t^{\nu_1+\delta_1}\Bigl{)} .
    \end{eqnarray}
 and therefore
    \[
         \sigma(\mathcal{F}_t, \mathcal{G}^{\tau_1+\delta_1}_t)\mid_{ A\cap \{\sigma\leq t\}}\; \subseteq \mathcal{F}_t^{\nu_1+ \delta_1},
    \]
i.e.
     \[
        \mathcal{F}_t^{\tau_1+ \delta_1}\mid_{ A\cap \{\sigma\leq t\}}\; \subseteq \mathcal{F}_t^{\nu_1+ \delta_1}.
    \]
        For general $i$ by substituting the sets  $\{ \tau_1+\delta_1 \leq s \}$ by $\{ \tau_{i+1}+\delta_{i+1} \leq s \}$ in (\ref{eigenschaft rech im ia}), we obtain that
        \[
            A\cap \{\sigma\leq t\}\cap \{ \tau_{i+1}+\delta_{i+1} \leq s  \} \in \sigma( \mathcal{F}_t,\, \mathcal{G}_t^{\nu_{i+1}+\delta_{i+1}})
        \]
and by induction
     \[
         \mathcal{F}_t^{\tau_1+\delta_1,...,\tau_{i+1}+\delta_{i+1}}\mid_{ A\cap \{\sigma\leq t\}}\; \subseteq  \mathcal{F}_t^{\nu_1+\delta_1,...,\nu_{i+1}+\delta_{i+1}}.
    \]
 \bigskip
Analogously we obtain
 \[
    \mathcal{F}_t^{\rho_1+\delta_1,...,\rho_i+\delta_i}\mid_{ A^c\cap \{\sigma\leq t\}}\; \subseteq  \mathcal{F}_t^{\nu_1+\delta_1,...,\nu_i+\delta_i}
 \]
and therefore
 \begin{eqnarray*}
    \{\nu_{i}+ \delta_{i} \leq t\} & = & \Bigr{(}A\cap \{ \tau_i+\delta_i \leq t \} \Bigl{)}\; \cup \;  \Bigr{(}A^c\cap \{ \rho_i+\delta_i \leq t \}\Bigl{)}\\
    & = & \Bigr{(}A\cap \{\sigma\leq t\}\cap \{ \tau_i+\delta_i \leq t \} \Bigl{)} \; \cup \;  \Bigr{(}A^c\cap  \{\sigma\leq t\}\cap  \{ \rho_i+\delta_i \leq t \}\Bigl{)}\\
    & \in & \sigma\Bigr{(} \mathcal{F}_t^{\tau_1+\delta_1,...,\tau_{i-1}+\delta_{i-1}}\mid_{ A\cap \{\sigma\leq t\}},\, \mathcal{F}_t^{\rho_1+\delta_1,...,\rho_{i-1}+\delta_{i-1}}\mid_{ A^c\cap \{\sigma\leq t\}}\Bigl{)}\\
    & \subseteq & \mathcal{F}_t^{\nu_1+\delta_1,...,\nu_{i-1}+\delta_{i-1}}.
\end{eqnarray*}
This proves that
$(\nu_1,...,\nu_n)\in\mathcal{S}_\sigma^{n}(\delta)$ and it holds that
 \begin{eqnarray*}\qquad \quad \;
 \E\Bigr{(}\sum_{i=1}^nY(\nu_i)\Big{|}\mathcal{F}_\sigma\Bigl{)}
                 =
 \max\Bigr{\{}\E\Bigr{(}\sum_{i=1}^nY(\tau_i)\Big{|}\mathcal{F}_\sigma\Bigl{)} ,\; \E\Bigr{(}\sum_{i=1}^nY(\rho_i)\Big{|}\mathcal{F}_\sigma\Bigl{)}\Bigl{\}}.
 \end{eqnarray*}
 This proves the first statement. For the second assertion note that
if $\text{range}(\tau_1)$ and $\text{range}(\rho_1)$ are countable, then so is $\text{range}(\nu_1)$, i.e. by the proof above $(\nu_1,....,\nu_n)\in\mathcal{S}_\sigma^{n}(\delta)_{\text{disc}}$.


\begin{thebibliography}{KQRM11}

\bibitem[Ben11a]{BenderDual2}
Christian Bender.
\newblock Dual pricing of multi-exercise options under volume constraints.
\newblock {\em Finance Stoch.}, 15(1):1--26, 2011.

\bibitem[Ben11b]{BenderDual}
Christian Bender.
\newblock Primal and dual pricing of multiple exercise options in continuous
  time.
\newblock {\em SIAM J. Finan. Math.}, 2:562--586, 2011.

\bibitem[BL97]{BL}
M.~Beibel and H.~R. Lerche.
\newblock A new look at optimal stopping problems related to mathematical
  finance.
\newblock {\em Statist. Sinica}, 7(1):93--108, 1997.
\newblock Empirical Bayes, sequential analysis and related topics in statistics
  and probability (New Brunswick, NJ, 1995).

\bibitem[BS02]{Borodin}
Andrei~N. Borodin and Paavo Salminen.
\newblock {\em Handbook of {B}rownian motion---facts and formulae}.
\newblock Probability and its Applications. Birkh{\"a}user Verlag, Basel,
  second edition, 2002.

\bibitem[BS06]{Bender}
Christian Bender and John Schoenmakers.
\newblock An iterative method for multiple stopping: convergence and stability.
\newblock {\em Adv. in Appl. Probab.}, 38(3):729--749, 2006.

\bibitem[CD08]{Car}
Ren{\'e} Carmona and Savas Dayanik.
\newblock Optimal multiple stopping of linear diffusions.
\newblock {\em Math. Oper. Res.}, 33(2):446--460, 2008.

\bibitem[CRS71]{CRS}
Y.~S. Chow, Herbert Robbins, and David Siegmund.
\newblock {\em Great expectations: the theory of optimal stopping}.
\newblock Houghton Mifflin Co., Boston, Mass., 1971.

\bibitem[CT08]{Touzi}
Ren{\'e} Carmona and Nizar Touzi.
\newblock Optimal multiple stopping and valuation of swing options.
\newblock {\em Math. Finance}, 18(2):239--268, 2008.

\bibitem[JRT04]{Jaillet}
Patrick Jaillet, Ehud~I. Ronn, and Stathis Tompaidis.
\newblock Valuation of commodity-based swing options.
\newblock {\em Management Science}, 50(7):900--921, 2004.

\bibitem[KQRM11]{KQR}
Magdalena Kobylanski, Marie-Claire Quenez, and Elisabeth Rouy-Mironescu.
\newblock {Optimal multiple stopping time problem.}
\newblock {\em Ann. Appl. Probab.}, 21(4):1365--1399, 2011.

\bibitem[KS88]{KS}
Ioannis Karatzas and Steven~E. Shreve.
\newblock {\em Brownian motion and stochastic calculus}, volume 113 of {\em
  Graduate Texts in Mathematics}.
\newblock Springer-Verlag, New York, 1988.

\bibitem[MH04]{Meinshausen}
N.~Meinshausen and B.~M. Hambly.
\newblock Monte {C}arlo methods for the valuation of multiple-exercise options.
\newblock {\em Math. Finance}, 14(4):557--583, 2004.

\bibitem[PS06]{PS}
Goran Peskir and Albert Shiryaev.
\newblock {\em Optimal stopping and free-boundary problems}.
\newblock Lectures in Mathematics ETH Z\"urich. Birkh\"auser Verlag, Basel,
  2006.

\bibitem[Tho95]{Thompson}
Andrew~C. Thompson.
\newblock Valuation of path-dependent contingent claims with multiple exercise
  decisions over time: The case of take-or-pay.
\newblock {\em J. Financial Quant. Anal.}, 30:271--293, 1995.

\end{thebibliography}
\end{document}